\documentclass[11pt,letterpaper]{article}
\usepackage{amsmath,amsfonts,amssymb,amsthm}
\usepackage{graphicx,xcolor}
\usepackage{fullpage}
\usepackage[font=small]{caption}
\usepackage{subcaption}
\usepackage{hyperref}

\newtheorem{lemma}{Lemma}
\newtheorem{theorem}{Theorem}

\newcommand{\MST}{{\rm MST}}
\newcommand{\geo}{{\rm geod}}
\newcommand{\conv}{{\rm conv}}
\newcommand{\diam}{{\rm diam}}

\title{Minimum Weight Connectivity Augmentation\\ for Planar Straight-Line Graphs\thanks{Research on this paper was supported in part by the NSF awards CCF-1422311 and CCF-1423615.}}

\author{Hugo A. Akitaya\thanks{Department of Computer Science, Tufts University, Medford, MA}
\and R. Inkulu\thanks{Department of Computer Science \& Engineering, Indian Institute of Technology Guwahati, India}
\and Torrie L. Nichols\thanks{Department of Mathematics, California State University Northridge, Los Angeles, CA}
\and Diane L. Souvaine\footnotemark[2]
\and Csaba D. T\'oth\footnotemark[2] \footnotemark[4]
\and Charles R. Winston\footnotemark[2]
}

\newtheorem{remark}{Lemma}
\newtheorem{corollary}{Corollary}

\begin{document}
\maketitle

\begin{abstract}
We consider edge insertion and deletion operations that increase the connectivity of a given planar straight-line graph (PSLG), while minimizing the total edge length of the output. We show that every connected PSLG $G=(V,E)$ in general position can be augmented to a 2-connected PSLG $(V,E\cup E^+)$ by adding new edges of total Euclidean length $\|E^+\|\leq 2\|E\|$, and this bound is the best possible. An optimal edge set $E^+$ can be computed in $O(|V|^4)$ time; however the problem becomes NP-hard when $G$ is disconnected.
Further, there is a sequence of edge insertions and deletions that transforms a connected PSLG $G=(V,E)$ into a planar straight-line cycle $G'=(V,E')$ such that $\|E'\|\leq 2\|\MST(V)\|$, and the graph remains connected with edge length below $\|E\|+\|\MST(V)\|$ at all stages.  These bounds are the best possible.
\end{abstract}

\section{Introduction}\label{sec:intro}

Connectivity augmentation is a classical problem in combinatorial optimization. Given a graph $G=(V,E)$ and a parameter $\tau\in \mathbb{N}$, add a set of new edges $E^+$ of minimum cardinality or weight such that the augmented graph $G'=(V,E\cup E^+)$ is $\tau$-connected (resp., $\tau$-edge-connected). Efficient algorithms are known for both connectivity and edge-connectivity augmentation over abstract graphs and constant $\tau$~\cite{Fra11,MV15,Veg11}.
In this paper we consider weighted connectivity augmentation for planar straight-line graphs (PSLGs). The vertices are points in Euclidean plane, the edges are noncrossing line segments between the corresponding vertices, and the weight of an edge is its Euclidean length.


The edge- and node-connectivity of a planar graph is at most 5 by Euler's theorem.
Further, not every PSLG can be augmented to a 3-connected (resp., 4-edge-connected) PSLG; see~\cite{HT13} for feasibility conditions. Finding the minimum \emph{number} of edges to augment a given PSLG to $\tau$-connectivity or $\tau$-edge-connectivity is NP-complete~\cite{RW12} for $2\leq \tau\leq 5$; the reduction requires the input graph $G$ to be disconnected (the NP-hardness claim for connected input~\cite[Corollary~2]{RW12} turned out to be flawed).
Worst case bounds are known for the most important cases: Every PSLG $G$ with $n$ vertices can be augmented to 2-edge-connectivity with at most $\lfloor (4n-4) / 3 \rfloor$ edges~\cite{ACL+15}; and at most $\lfloor (2n-2) / 3 \rfloor$ new edges if $G$ is already connected~\cite{Tot12}. At most $b-1$ suffice for 2-connectivity, where $b$ is the number of 2-blocks in $G$~\cite{AGF+08}. All these bounds are the best possible.

\smallskip\noindent{\bf
Our results.}
We show that every connected PSLG $G=(V,E)$ with $n$ vertices in general position can be augmented to a 2-connected PSLG $G'=(V,E\cup E^+)$ by adding new edges of total Euclidean length $\|E^+\|\leq 2\|E\|$ (Section~\ref{sec:augment}). A set $E^+$ that minimizes $\|E^+\|$ can be computed in $O(|V|^4)$ time (Section~\ref{sec:DP}); however the problem becomes NP-hard when $G$ is disconnected (Section~\ref{sec:hardness}). Further, there is a sequence of edge insertions and deletions that transforms a connected PSLG $G$ into a planar straight-line cycle $G'=(V,E')$ such that the graph at all stages remains connected with the sum of Euclidean edge lengths below $\|E\|+\|\MST(V)\|$, and, at termination, $\|E'\|\leq 2\|\MST(V)\|$, where $\MST(V)$ denotes the Euclidean minimum spanning tree of $V$; these bounds are the best possible (Section~\ref{sec:dynamic}). Our proof is constructive, and yields a polynomial-time algorithm for computing such sequences.

\smallskip\noindent{\bf
Related previous work.}
Biconnectivity augmentation over \emph{planar} graphs (where no embedding of $G$ is given) is also NP-complete~\cite{KB91}. Over planar graphs \emph{with fixed combinatorial embedding}, biconnectivity augmentation remains NP-hard for disconnected graphs; but there is a near-linear time algorithm when the input is already connected~\cite{GMZ09}.
Frederickson and Ja'Ja'~\cite{FJ81} show that the \emph{weighted} augmentation of a tree to be 2-connected or 2-edge-connected (without planarity constraint) is NP-complete even if the weights are restricted to $\{1,2\}$.
The problem is APX-hard, and  breaking the approximation ratio of 2 is a major open problem~\cite{KZ16}. In the \emph{geometric} setting,
we show (in Section~\ref{sec:DP}) that the minimum-weight augmentation of a planar straight-line tree to a 2-connected (2-edge-connected) PSLGs  can be computed efficiently.

The length of the edges in planar connectivity augmentation was studied only recently in the context of wireless networks. Given a PLSG $G = (V,E)$ where the vertices induce a 2-edge-connected unit disk graph, Dobrev et al.~\cite{DKK+12} compute a 2-edge-connected PSLG by adding edges of length at most 2. Kranakis et al.~\cite{KKM+12} studied the combined problem of adding the minimum \emph{number} of edges of \emph{bounded length}: A 2-edge-connected augmentation is possible such that $|E^+|$ is at most the number of bridges in $G$ and $\max_{e'\in E^+} \|e'\|\leq 3\max_{e\in E}\|e\|$. However, finding the minimum number of new edges of bounded length is NP-hard.

\section{Preliminaries}\label{sec:prelim}

Let $G=(V,E)$ be a planar straight-line graph (PSLG), where $V$ is a set of $n$ points in the plane,  no three of which are collinear, and $E$ is a set of open line segments between pairs of points in $V$. The length of an edge $uv$, denoted $\|uv\|$, is the Euclidean distance between $u$ and $v$; the total length of the edges is $\|E\|=\sum_{e\in E} \|e\|$. Denote by $F$ the set of faces of $G$. The \emph{faces} of $G$ are the connected components of the complement of all vertices and edges of $G$, that is, $\mathbb{R}^2\setminus (V\cup \bigcup_{e\in E}e)$. If $G$ is connected, then every bounded face is simply connected.

A \emph{walk} is an alternating sequence of points (vertices) and line segments (edges) whose consecutive elements are incident, and hence it is uniquely described by the sequence of its vertices $w=(p_0,\ldots ,p_t)$. A walk is \emph{closed} if $p_0=p_t$. A walk is called a \emph{path} if no vertex appears more than once. Every face $f\in F$ defines a closed walk $(p_0,\ldots ,p_t)$ that contain all edges on the boundary of $F$, called \emph{facial walk}, where every edge $p_{i-1}p_i$ is incident to the face $f$, and consecutive edges in the path, $p_{i-1}p_i$ and $p_ip_{i+1}$, are also consecutive in the counterclockwise rotation of all edges incident to $p_i$ (see Fig.~\ref{fig:2}(a)). Note that every edge $e\in E$ occurs twice in the facial walks of the faces of $G$: A \emph{cut vertex} (resp., \emph{bridge}) occurs twice in some facial walk.

A walk $p=(p_0,\ldots ,p_t)$ is \emph{convex} if $0<\angle p_{i-1}p_ip_{i+1} <\pi$ for $i=1,\ldots , t-1$ (for closed walks, $i=1,\ldots , t$), where $\angle p_{i-1}p_ip_{i+1}$ is the measure of the minimum counterclockwise angle that rotates the ray $\overrightarrow{p_ip_{i-1}}$ into $\overrightarrow{p_ip_{i+1}}$.
A vertex of $G$ is called \emph{convex} if $0<\angle p_{i-1}p_ip_{i+1} <\pi$ for each pair $p_{i-1}p_i$ and $p_ip_{i+1}$ of consecutive edges in the counterclockwise rotation of edges incident to $p_i$; and \emph{reflex} otherwise.
A convex walk $p=(p_0,\ldots , p_t)$ is \emph{safe} if not all of its vertices are collinear, and the vertices lie on the boundary of the convex hull of $\{p_0,\ldots , p_t\}$ with the possible exception of the first or last vertex.

Let $p=(p_0,\ldots , p_t)$ be a walk contained in the boundary walk of a face $f\in F$.
The shortest walk from $p_0$ to $p_t$ homotopic to $p$, denoted $\geo(p)$, is called the \emph{geodesic} between $p_0$ and $p_t$. For a walk $p$, let $\gamma_p:[0,1]\rightarrow \mathbb{R}^2$ be a piecewise linear arc from $p_0$ to $p_t$ that traverses the edges of $p$ in the given order. A \emph{homotopy} between two walks, $p=(p_0,\ldots , p_t)$ and $q=(q_0,\ldots , q_{t'})$, is a continuous function $h: [0,1]\times[0,1]\rightarrow \mathbb{R}^2$ such that $h(0,\cdot) = \gamma_p(\cdot)$, $h(1,\cdot) = \gamma_q(\cdot)$, $h(\cdot,0) = p_0=q_0$, $h(\cdot, 1) = p_t=q_{t'}$, $h(a,b)\in \bigcup_{f\in F}f$ for all $(a,b)\in (0,1)\times (0,1)$. Intuitively, the walk $p$ can be continuously deformed into $q$ in the face $f$.
The walks $p$ and $q$ are \emph{homotopic} if such a homotopy exists.

\begin{figure}[htb]
\centering
\def\svgwidth{.8\columnwidth}
\begingroup%
  \makeatletter%
  \providecommand\color[2][]{%
    \errmessage{(Inkscape) Color is used for the text in Inkscape, but the package 'color.sty' is not loaded}%
    \renewcommand\color[2][]{}%
  }%
  \providecommand\transparent[1]{%
    \errmessage{(Inkscape) Transparency is used (non-zero) for the text in Inkscape, but the package 'transparent.sty' is not loaded}%
    \renewcommand\transparent[1]{}%
  }%
  \providecommand\rotatebox[2]{#2}%
  \ifx\svgwidth\undefined%
    \setlength{\unitlength}{429.43066406bp}%
    \ifx\svgscale\undefined%
      \relax%
    \else%
      \setlength{\unitlength}{\unitlength * \real{\svgscale}}%
    \fi%
  \else%
    \setlength{\unitlength}{\svgwidth}%
  \fi%
  \global\let\svgwidth\undefined%
  \global\let\svgscale\undefined%
  \makeatother%
  \begin{picture}(1,0.25453794)%
    \put(0,0){\includegraphics[width=\unitlength,page=1]{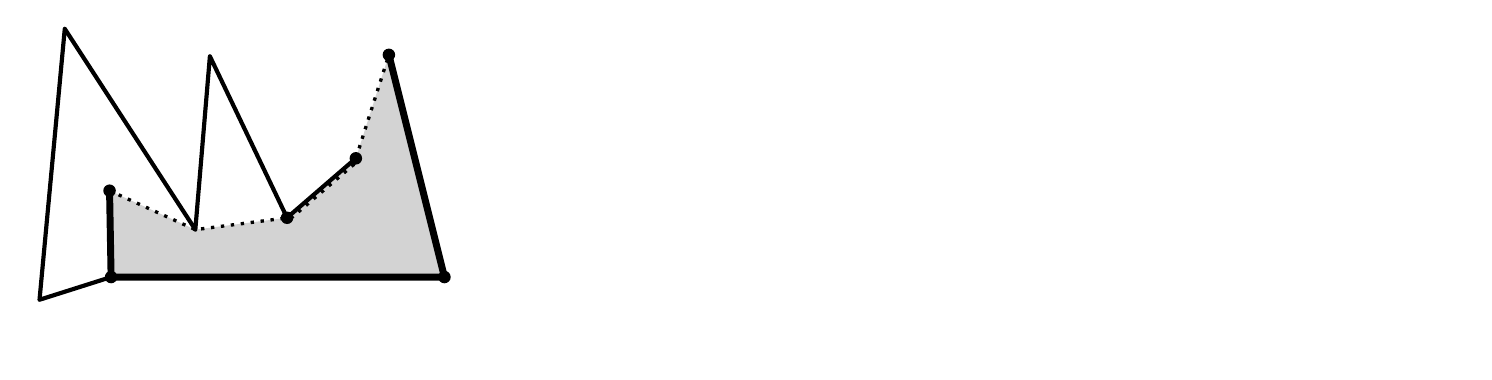}}%
    \put(0.28871011,0.04562865){\color[rgb]{0,0,0}\makebox(0,0)[lb]{\smash{$p_1$}}}%
    \put(0,0){\includegraphics[width=\unitlength,page=2]{geod-example.pdf}}%
    \put(0.12182641,0.00579987){\color[rgb]{0,0,0}\makebox(0,0)[lb]{\smash{(a)}}}%
    \put(0,0){\includegraphics[width=\unitlength,page=3]{geod-example.pdf}}%
    \put(0.48033978,0.00580222){\color[rgb]{0,0,0}\makebox(0,0)[lb]{\smash{(b)}}}%
    \put(0,0){\includegraphics[width=\unitlength,page=4]{geod-example.pdf}}%
    \put(0.82863433,0.00580222){\color[rgb]{0,0,0}\makebox(0,0)[lb]{\smash{(c)}}}%
    \put(0.2537617,0.23149408){\color[rgb]{0,0,0}\makebox(0,0)[lb]{\smash{$p_0$}}}%
    \put(0.0668696,0.03822856){\color[rgb]{0,0,0}\makebox(0,0)[lb]{\smash{$p_2$}}}%
    \put(0.05432268,0.13912942){\color[rgb]{0,0,0}\makebox(0,0)[lb]{\smash{$p_3$}}}%
    \put(0.60850681,0.04221873){\color[rgb]{0,0,0}\makebox(0,0)[lb]{\smash{$p_1$}}}%
    \put(0.51319131,0.24042385){\color[rgb]{0,0,0}\makebox(0,0)[lb]{\smash{$p_0$}}}%
    \put(0.38191831,0.03798715){\color[rgb]{0,0,0}\makebox(0,0)[lb]{\smash{$p_2$}}}%
    \put(0.43719944,0.15077211){\color[rgb]{0,0,0}\makebox(0,0)[lb]{\smash{$p_3$}}}%
    \put(0.81877382,0.24583347){\color[rgb]{0,0,0}\makebox(0,0)[lb]{\smash{$p_1$}}}%
    \put(0.68479637,-0.00034599){\color[rgb]{0,0,0}\makebox(0,0)[lb]{\smash{$p_0$}}}%
    \put(0.96493903,0.04462578){\color[rgb]{0,0,0}\makebox(0,0)[lb]{\smash{$p_2$}}}%
    \put(0.75089298,0.0516515){\color[rgb]{0,0,0}\makebox(0,0)[lb]{\smash{$p_3$}}}%
    \put(0.8487234,0.13190323){\color[rgb]{0,0,0}\makebox(0,0)[lb]{\smash{$p_4$}}}%
  \end{picture}%
\endgroup%
\caption{(a) A convex path $p=(p_0,\ldots ,p_3)$ with vertices in convex position.
(b)  A convex path $p=(p_0,\ldots , p_3)$, where $p_0\in \conv(p_1,p_2,p_3)$.
(c) A convex path $p=(p_0,\ldots , p_4)$ with four edges, where $\geo(p)$ contains edge $p_0p_1$.}
\label{fig:1}
\end{figure}

Given a walk $p=(p_0,\ldots , p_t)$ and polygonal environment $\bigcup_{f\in F}f$ with $n$ vertices,
the geodesic can be computed in $O(n\log n)$ time~\cite{Bes03}.
It is known~\cite{EKL06,HS94} that all interior vertices of $\geo(p)$ are reflex vertices of $G$;
and if $p$ is a convex chain, then so is $\geo(p)$ (it may be a straight-line segment). Geodesics play a crucial role in our worst-case bounds, since $\|\geo(p)\|\leq \|p\|$ by definition, and the edges of $\geo(p)$ do not cross any existing edges of $G$. We show the following.

\begin{lemma}\label{lem:1}
Let $G=(V,E)$ be a PSLG, and let $p=(p_0,\ldots, p_t)$ be a safe convex walk contained in some facial walk of $G$.
Then $\geo(p)$ is a simple path that does not contain any vertices of $p$ except for $p_0$ and $p_t$ at its endpoints.
\end{lemma}
\begin{proof}
Let $H$ be the convex hull of $\{p_0,\ldots, p_t\}$. By assumption, the vertices of $p$ lie on $\partial H$ with the possible exception of $p_0$ and $p_t$, which may lie in the interior of $H$. By construction, $\geo(p)$ lies in ${\rm int}(H)$ with the possible exception of its endpoints $p_0$ and $p_t$. Hence, $\geo(p)$ does not contain any interior vertex of $p$.
\end{proof}

\begin{lemma}\label{lem:1+}
Let $G=(V,E)$ be a PSLG, and let $p=(p_0,\ldots, p_t)$ be a convex walk contained in some facial walk of $G$.
If $\geo(p_1,\ldots, p_t)$ is a simple path that contains none of the vertices $p_0$, $p_2,\ldots, p_{t-1}$,
then $\geo(p)$ is also a simple path that does not contain any of the vertices $p_1,\ldots, p_{t-1}$.
\end{lemma}
\begin{proof}
By Lemma~\ref{lem:1}, $C=(p_1,\ldots, p_t)\cup \geo(p_1,\ldots, p_t)$ is a simple cycle. Let $p_1r$ be the first edge of $\geo(p_1,\ldots, p_t)$. Since $\geo(p_1,\ldots, p_t)$ is homotopic to $(p_1,\ldots , p_t)$, the edge $p_1r$ lies in the angular domain $\angle p_0p_1p_2$, and $r\not\in \{p_0,p_2\}$. Consequently, the edge $p_0p_1$ lies in the exterior of the cycle $C$. Note that $\geo(p)=\geo(p_0,p_1,r)\cup \geo(r,p_1,\ldots, p_t)$.

Since $\geo(p_1,\ldots, p_t)$ is a convex chain, the interior of triangle $\Delta(p_0,p_1,r)$ lies in the exterior of $C$. On the other hand, $\geo(p_0,p_1,r)$ lies in $\Delta(p_0,p_1,r)$, and so it is disjoint from the vertices $p_2,\ldots, p_t$. We conclude that $\geo(p)$ is a simple path that does not contain any of the vertices $p_1,\ldots, p_{t-1}$.
\end{proof}

\section{Bounds on the Sum of Edge Lengths}\label{sec:augment}

Let $G$ be a PSLG with no three collinear vertices. Denote by ${\mathcal P}=\mathcal{P}(G)$ the set of maximal convex walks contained in the facial walks of $G$ (note that ${\mathcal P}$ can be computed with a graph traversal in $O(|E|)$ time). Partition ${\mathcal P}$ into three subsets: ${\mathcal P}_0$ contains the convex walks that consist of a single edge; ${\mathcal P}_1$ contains the closed convex walks $(p_0,\ldots, p_t)$, i.e., $p_0=p_t$; and ${\mathcal P}_2$ contains all open convex walks of two or more edges. We define a \emph{dual graph} $D$ where the nodes correspond to the convex walks in  ${\mathcal P}_1\cup {\mathcal P}_2$, and two nodes are adjacent if and only if the corresponding convex chains share an edge in $E$.

\begin{lemma}\label{lem:2}
Let $G=(V,E)$ be a connected PSLG with $|V|\ge 3$. Then
\begin{enumerate}\itemsep -2pt
\item[{\rm (a)}] every edge in $E$ is part of a convex chain in ${\mathcal P}_1\cup {\mathcal P}_2$,
\item[{\rm (b)}]  the dual graph $D$ is connected,
\end{enumerate}
\end{lemma}
\begin{proof}
	(a)  Let $ab\in E$. Since $G$ has at least three vertices and is connected, we may assume that $b$ is incident to two or more edges. Assume that $bc^-$ and $bc^+$ are the edges preceding and following $ba$ in the clockwise rotation of edges around $b$ respectively (possibly $c^-=c^+$). Then the boundary walks of the faces incident to $b$ contain the paths $(c^-,b,a)$ and $(a,b,c^+)$. Since $b$ is the apex of at most one reflex angle, we may assume that $\angle c^-ba$ or $\angle abc^+$ is convex; and so $ab$ is part of a maximal convex walk of 2 or more edges. This walk is in either ${\mathcal P}_1$ or ${\mathcal P}_2$.
	
	(b) For every vertex $v\in V$, every angle $\angle uvw$ formed by consecutive incident edges $uv$ and $vw$ is convex with the possible exception of one reflex angle (since the sum of these angles is $2\pi)$. Consequently, the walks $(u,v,w)$ where $\angle uvw<\pi$ are part of distinct convex chains in ${\mathcal P}_1\cup {\mathcal P}_2$  that induce a path or a cycle in the dual graph. Now consider two chains $p_1,p_2\in {\mathcal P}_1\cup {\mathcal P}_2$, and two arbitrary edges $e_1$ and $e_2$ from them.
	Since $G$ is connected, there is a path $q=(q_0,\ldots , q_t)$ such that $e_1=q_0q_1$ and $e_2=q_{t-1}q_t$. Since every two consecutive edges of $q$ are part of convex cycles in the same component of $D$, the chains $p_1$ and $p_2$ are also part of the same component of $D$.
\end{proof}

The proof of Lemma~\ref{lem:2} is provided in the Appendix. When we modify a given PSLG with edge insertion operations, we prove the following.

\begin{theorem}\label{thm:1}
Let $G=(V, E)$ be a connected PSLG with $|V|\geq 3$ and
no three collinear vertices. Then $G$ can be augmented to a 2-edge-connected PSLG $G'=(V,E\cup E^+)$ such that $\|E^+\| \leq 2\|E\|$, and this bound is the best possible.
\end{theorem}
\begin{proof}
We prove the upper bound constructively, augmenting a connected PSLG $G=(V,E)$ incrementally into a 2-edge-connected PSLG $G'=(V,E\cup E^+)$. Then decompose every convex walk in ${\mathcal P}$ of two or more edges into edge-disjoint convex walks of two or three edges. Let ${\mathcal C}$ be the set of all resulting convex paths of two or three edges (that is, we discard convex walks that consist of a singe edge and convex cycles of 3 edges). See Fig.~\ref{fig:2}(b) for an illustration.

\begin{figure}[h]
\centering
\def\svgwidth{\columnwidth}
\begingroup%
  \makeatletter%
  \providecommand\color[2][]{%
    \errmessage{(Inkscape) Color is used for the text in Inkscape, but the package 'color.sty' is not loaded}%
    \renewcommand\color[2][]{}%
  }%
  \providecommand\transparent[1]{%
    \errmessage{(Inkscape) Transparency is used (non-zero) for the text in Inkscape, but the package 'transparent.sty' is not loaded}%
    \renewcommand\transparent[1]{}%
  }%
  \providecommand\rotatebox[2]{#2}%
  \ifx\svgwidth\undefined%
    \setlength{\unitlength}{828.41708984bp}%
    \ifx\svgscale\undefined%
      \relax%
    \else%
      \setlength{\unitlength}{\unitlength * \real{\svgscale}}%
    \fi%
  \else%
    \setlength{\unitlength}{\svgwidth}%
  \fi%
  \global\let\svgwidth\undefined%
  \global\let\svgscale\undefined%
  \makeatother%
  \begin{picture}(1,0.19059586)%
    \put(0,0){\includegraphics[width=\unitlength,page=1]{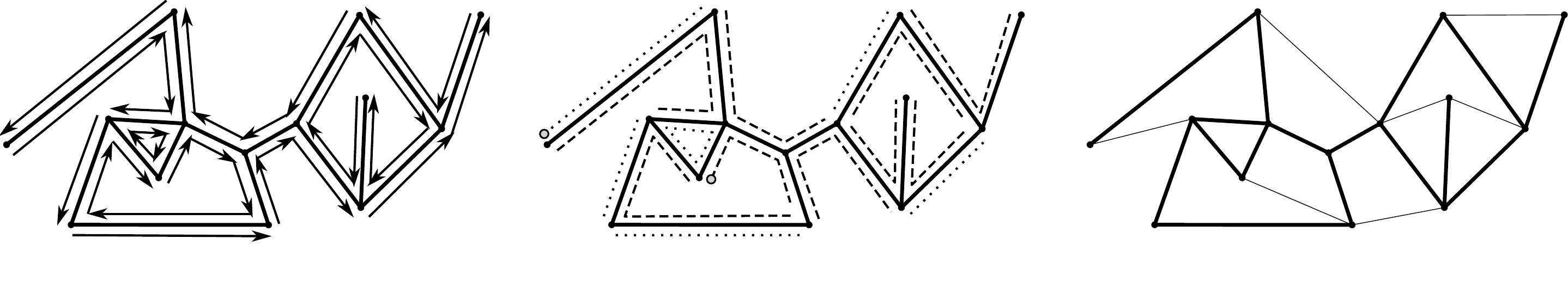}}%
    \put(0.1414397,0.00392041){\color[rgb]{0,0,0}\makebox(0,0)[lb]{\smash{(a)}}}%
    \put(0.49326912,0.00392041){\color[rgb]{0,0,0}\makebox(0,0)[lb]{\smash{(b)}}}%
    \put(0.84646753,0.00392041){\color[rgb]{0,0,0}\makebox(0,0)[lb]{\smash{(c)}}}%
    \put(0,0){\includegraphics[width=\unitlength,page=2]{edge-alg.pdf}}%
  \end{picture}%
\endgroup%
\caption{(a) A PSLG $G=(V,E)$ is its boundary walks.
 (b) Dashed lines indicate convex paths of two or three edges in ${\mathcal C}$;
 dotted lines indicate maximal convex paths of a single edge, and triangles.
 (c) The 2-edge-connected PSLG $G'=(V,E\cup E^+)$ produced by our algorithm.}
\label{fig:2}
\end{figure}

For each convex path $p\in \mathcal{C}$, augment $G$ with the edges of $\geo(p)$ (refer to Fig.~\ref{fig:2}(c)), and denote by $G'$ the resulting graph. Note that $G'$ is 2-edge-connected since every edge in $E$ is part of a cycle by Lemma~\ref{lem:2}(a): Each cycle is either a triangle in $G$, or a cycle $p\cup \geo(p)$ for some
$p\in {\mathcal C}$; and every edge in $E^+$ is part of a cycle by construction.

Next, we derive an upper bound for $\|E^+\|$. Every $e\in E$ appears twice in the boundary walks of the faces of $G$, and so it appears in at most two convex paths in $\mathcal C$. By definition, $\|\geo(p)\|\leq \|p\|$ for every $p\in {\mathcal C}$. Overall, we have
$$\|E^+\|=
\sum_{p\in \mathcal{C}} \|\geo(p)\| \leq
\sum_{p\in \mathcal{C}}\|p\|\leq
\sum_{p\in \mathcal{P}}\|p\|=
2\|E\|.$$

\begin{figure}[h]
\centering
\def\svgwidth{.3\columnwidth}
\begingroup%
  \makeatletter%
  \providecommand\color[2][]{%
    \errmessage{(Inkscape) Color is used for the text in Inkscape, but the package 'color.sty' is not loaded}%
    \renewcommand\color[2][]{}%
  }%
  \providecommand\transparent[1]{%
    \errmessage{(Inkscape) Transparency is used (non-zero) for the text in Inkscape, but the package 'transparent.sty' is not loaded}%
    \renewcommand\transparent[1]{}%
  }%
  \providecommand\rotatebox[2]{#2}%
  \ifx\svgwidth\undefined%
    \setlength{\unitlength}{147.6bp}%
    \ifx\svgscale\undefined%
      \relax%
    \else%
      \setlength{\unitlength}{\unitlength * \real{\svgscale}}%
    \fi%
  \else%
    \setlength{\unitlength}{\svgwidth}%
  \fi%
  \global\let\svgwidth\undefined%
  \global\let\svgscale\undefined%
  \makeatother%
  \begin{picture}(1,0.18503062)%
    \put(0,0){\includegraphics[width=\unitlength,page=1]{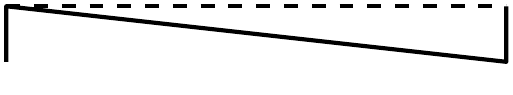}}%
    \put(1.00869647,0.10597043){\color[rgb]{0,0,0}\makebox(0,0)[lb]{\smash{$\varepsilon$}}}%
    \put(0,0){\includegraphics[width=\unitlength,page=2]{lowerbound.pdf}}%
    \put(-0.04272756,-0.02722992){\color[rgb]{0,0,0}\makebox(0,0)[lb]{\smash{$p_1$}}}%
    \put(0.46769892,-0.02384202){\color[rgb]{0,0,0}\makebox(0,0)[lb]{\smash{$1$}}}%
    \put(-0.0440314,0.22161412){\color[rgb]{0,0,0}\makebox(0,0)[lb]{\smash{$p_2$}}}%
    \put(0.96983538,-0.0114394){\color[rgb]{0,0,0}\makebox(0,0)[lb]{\smash{$p_3$}}}%
    \put(0.96983538,0.22090027){\color[rgb]{0,0,0}\makebox(0,0)[lb]{\smash{$p_4$}}}%
    \put(-0.04964884,0.09439476){\color[rgb]{0,0,0}\makebox(0,0)[lb]{\smash{$\varepsilon$}}}%
  \end{picture}%
\endgroup%
\caption{A PSLG $G=(V,E)$ with three solid edges, where the augmentation to 2-edge-connectivity requires the addition of the dashed edges.}
\label{fig:3}
\end{figure}

We now show a matching lower bound.
For every $\varepsilon>0$, let $G_\varepsilon$ be defined on four vertices $p_1=(0,0)$, $p_2=(0,\varepsilon)$, $p_3=(1,0)$, and $p_4=(1,\varepsilon)$ with edge set $E=\{p_1p_2, p_2p_3,p_3p_4\}$; refer to  Figure \ref{fig:3}.
Since $p_1$ and $p_4$ are leaves in $G$,
and $p_1p_4$ would cross $p_2p_3$, both $p_1p_3$ and $p_2p_4$ have to be added.
We have $\lim_{\varepsilon \rightarrow 0}\|E^+\|/\|E\|=2$, and the ratio $\|E^+\|/\|E\|\leq 2$ is the best possible.
\end{proof}

We strengthen Theorem~\ref{thm:1} to vertex-connectivity.
\begin{theorem}\label{thm:2}
Let $G=(V, E)$ be a connected PSLG with $|V|\geq 3$ and
no three collinear vertices. Then $G$ can be augmented to a 2-connected PSLG $G=(V, E\cup E^+)$ such that $\|E^+\| \leq 2\|E\|$, and this bound is the best possible.
\end{theorem}
\begin{proof}
We prove the upper bound constructively. Consider the convex walk in ${\mathcal P}_1\cup {\mathcal P}_2$, defined above. We augment $G$ into $G'=(V,E\cup E^+)$ such that the vertex set of each convex walk in ${\mathcal P}_1\cup {\mathcal P}_2$ induces a 2-connected subgraph in $G'$; and every new edge in $E^+$ is part of one of these subgraphs.
Note that this implies that $G'$ is 2-connected: By Lemma~\ref{lem:2}(a), every vertex is part of a 2-connected subgraph; if two 2-connected subgraphs share two vertices, then their union is 2-connected. By Lemma~\ref{lem:2}(b) the union of the subgraphs induced by the convex chains is 2-connected. In the remainder of the proof, we consider a single convex walks $p\in {\mathcal P}_1\cup {\mathcal P}_2$.

\noindent \textbf{Case~1: $p=(p_0,\ldots, p_t)\in {\mathcal P}_1$.}
If the vertices $p_0,\ldots, p_{t-1}$ are distinct, then $p$ is a cycle, and all vertices of the walk are part of a 2-connected component. Otherwise $p_1=p_{t-1}$. In this case, $(p_1,\ldots , p_{t-1})$ forms a convex polygon, whose interior contains $p_0=p_t$ but no other vertices. Consequently, $t\geq 4$, and the only cut vertex along the walk is $p_1$. Add the edge $p_0p_2$, where $\|p_0p_2\|\leq \|p_0p_1\|+\|p_1p_2\|\leq \|p\|$ by the triangle inequality (Fig.~\ref{fig:2-con-upperbound}(a)). As a result, the vertices of $p$ induce a 2-connected subgraph.

\noindent \textbf{Case~2: $p=(p_0,\ldots, p_t)\in {\mathcal P}_2$.} We decompose $p$ into edge-disjoint walks recursively as follows. If $\geo(p)$ does not contain any interior vertex of $p$, then we are done. Otherwise, let $H$ be the convex hull of $\{p_0,\ldots, p_t\}$ and let $p'=(p_i,\ldots, p_j)$ be the subchain along $\partial H$. By Lemma~\ref{lem:1}, the set of interior vertices of $\geo(p')$ does not contain any vertex of $p'$.
Starting from $p'$, successively append the edges of $p$ proceeding $p_i$ or following $p_j$ while the path $p'$ maintains the property that $p'$ contains no interior vertices of $\geo(p')$. Then recurse on any prefix or suffix path in $p\setminus p'$. We obtain a decomposition of $p$ into subpaths $p'$ such that
$p'\cup \geo(p')$ is a simple cycle (Fig.~\ref{fig:2-con-upperbound}(b)).
By Lemma~\ref{lem:1+}, any two such consecutive paths share two vertices. Consequently, the union of the cycles $p'\cup \geo(p')$ is a 2-connected graph.

Analogously to the proof of Theorem~\ref{thm:1}, we have $\|E^+\|\leq 2\|E\|$, and the same lower bound construction
shows that this bound is the best possible.
\end{proof}

\begin{figure}[t]
	\centering
	\def\svgwidth{.9\columnwidth}
\begingroup%
  \makeatletter%
  \providecommand\color[2][]{%
    \errmessage{(Inkscape) Color is used for the text in Inkscape, but the package 'color.sty' is not loaded}%
    \renewcommand\color[2][]{}%
  }%
  \providecommand\transparent[1]{%
    \errmessage{(Inkscape) Transparency is used (non-zero) for the text in Inkscape, but the package 'transparent.sty' is not loaded}%
    \renewcommand\transparent[1]{}%
  }%
  \providecommand\rotatebox[2]{#2}%
  \ifx\svgwidth\undefined%
    \setlength{\unitlength}{267.60000634bp}%
    \ifx\svgscale\undefined%
      \relax%
    \else%
      \setlength{\unitlength}{\unitlength * \real{\svgscale}}%
    \fi%
  \else%
    \setlength{\unitlength}{\svgwidth}%
  \fi%
  \global\let\svgwidth\undefined%
  \global\let\svgscale\undefined%
  \makeatother%
  \begin{picture}(1,0.2037577)%
    \put(0,0){\includegraphics[width=\unitlength,page=1]{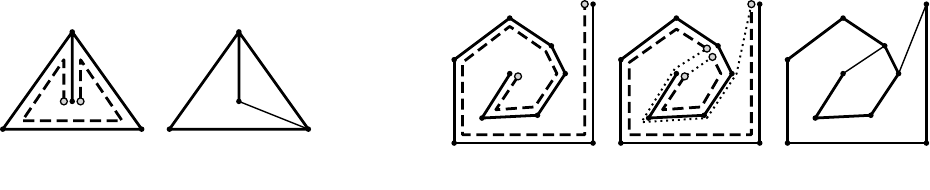}}%
    \put(0.14552631,0.0058292){\color[rgb]{0,0,0}\makebox(0,0)[lb]{\smash{(a)}}}%
    \put(0.72101208,0.0058292){\color[rgb]{0,0,0}\makebox(0,0)[lb]{\smash{(b)}}}%
  \end{picture}%
\endgroup%
	\caption{Examples for (a) case 1 and (b) case 2. }
	\label{fig:2-con-upperbound}
\end{figure}

\section{Algorithms for Connectivity Augmentation from 1 to 2}\label{sec:DP}

Let $F$ be a face of a PSLG $G$, and let $W_F=(p_0, \ldots, p_n)$, $p_0=p_{n}$ be the (closed) facial walk of $F$.
We define the graph $G_F=(V_F,E_F)$, where $V_F$ is the set of vertices in $W_F$ and $E_F=\{\{p_i,p_{i+1}\} \text{ where } i\in\{1,\ldots,n\}\}$. Given a walk $W_{s,t}=(p_s,\ldots,p_t)$ contained in $W_F$, a vertex $v_c$ is a \emph{cut vertex relative to $W_{s,t}$} if it appears more than once in $W_{s,t}$. For every pair $1\leq i<j\leq n$, we introduce a weight function, corresponding to the feasibility of an edge $p_ip_j$. Let $f(i,j)=\|p_ip_j\|$ if the line segment $p_ip_j$ does not cross any edge in $E_F$, lies in the face $F$ and in the wedges $\angle p_{i-1}p_ip_{i+1}$ and $\angle p_{j-1}p_jp_{j+1}$; and let $f(i,j)=\infty$ otherwise. Note that two feasible edges, $\{p_i,p_j\}$ and $\{p_{i'}, p_{j'}\}$, do not cross if their indices do not interleave (i.e., if $1\leq i<j\leq i'<j'\leq n$).

We present a dynamic programming algorithm $A$ that finds a set $E^+$ of edges of minimum total weight such that $(V_F,E_F\cup E^+)$ is a 2-connected PSLG.
We call an optimal solution $E_{OPT}$.
\begin{remark}\label{rem:separation}
Every edge in $E_{OPT}$ lies in the face $F$. Indeed, suppose $(V_F,E_F\cup E^+)$ is a 2-connected PSLG and $\{p_i,p_j\}\in E^+$ is outside $F$. Then $p_i$ and $p_j$ are part of a simple cycle formed by some edges of $G_F$, and $(V_F,E_F\cup (E^+\setminus\{p_i,p_j\}))$ is also 2-connected, showing that $E^+$ is not optimal.
\end{remark}
By Remark~\ref{rem:separation}, each face of the input can be treated independently. If we insert an edge in face $F$, then it decomposes $F$ into two faces $F_1$ and $F_2$ that can be considered independently. However, defining subproblems in terms of faces might generate an exponential number of subproblems. Instead we define our subproblems in term of continuous intervals of the facial walk of $F$.

\smallskip
We characterize an optimal solution $E_{OPT}$ for $G_F$ in terms of local properties of the subproblems $W_{s,t}$.
Let $p_c$ be a cut vertex with respect to a walk $W_{s,t}$. The vertices between two consecutive occurrences of $p_c$ in $W_{s,t}$ are called \emph{descendants} of $p_c$. A \emph{non-descendant} of $p_c$ in $W_{i,j}$ is a vertex in $W_F$ that is neither a descendant in $W_{i,j}$ nor equal to $p_c$. The \emph{$k$-th group} of descendants is the set of vertices between the $k$-th and $(k+1)$-st occurrence of $p_c$.
A set $E'$ of feasible edges \emph{satisfies} a group if there is a cycle in the graph $G'_F=(V_F\cup E_F\cup E')$ that contains a descendant in that group and a non-descendant of $p_c$; and $E'$ satisfies a cut vertex $p_c$ if it satisfies all of its groups. If $(V_F,E_F\cup E^+)$ is a 2-connected PSLG, then $E^+$ satisfies all cut vertex in $W_{1,n}$.
Indeed, suppose $E^+$ does not satisfy a vertex $p_c$, then the deletion of $p_c$ would disconnect one of its groups of descendants from the rest of $G_F$, hence $p_c$ would be a cut vertex in $(V_F,E_F\cup E^+)$.

Let $C[s,t]$, $s\le t$, be the minimum weight of an edge set $E'$ that satisfies all groups of all cut vertices relative to $W_{s,t}$ and such that $\{p_i,p_j\}\in E'$, $i,j\in\{s,\ldots,t\}$ and $(V_F,E_F\cup E')$ is a PSLG.
Algorithm $A$ uses the following recursive relation to compute subproblems $C[s,t]$ and returns $C[1,n]$.

\begin{enumerate}
\item [(i)] If $W_{s,t}$ does not contain any cut vertex relative to $W_{s,t}$, then $C[s,t]=0$.

\item [(ii)] If $p_s=p_t$ and $s\neq t$, then $C[s,t]=\infty$.

\item [(iii)] If $p_s$ is not a cut vertex relative to $W_{s,t}$, then $C[s,t]=\min\{C[s+1,t]$,\\ and $\min_{k\in\{s+2,\ldots,t-1\}}\{ C[s,k]+C[k,t]+f(s,k)\}\}$.

\item [(iv)] If $p_s$ is a cut vertex relative to $W_{s,t}$, then let\\
    $X=\{\text{descendants of }p_s\text{ in }W_{s,t}\}\times
    \{\text{non-descendants of }p_s\text{ in }W_{s,t}\}$.\\
Set $C[s,t]=\min_{ (p_i,p_j)\in X}\{ C[s,i]+C[i,j]+C[j,t]+f(i,j)\}$.
\end{enumerate}


\begin{figure}[t]
\centering
	\def\svgwidth{0.8\columnwidth}
\begingroup%
  \makeatletter%
  \providecommand\color[2][]{%
    \errmessage{(Inkscape) Color is used for the text in Inkscape, but the package 'color.sty' is not loaded}%
    \renewcommand\color[2][]{}%
  }%
  \providecommand\transparent[1]{%
    \errmessage{(Inkscape) Transparency is used (non-zero) for the text in Inkscape, but the package 'transparent.sty' is not loaded}%
    \renewcommand\transparent[1]{}%
  }%
  \providecommand\rotatebox[2]{#2}%
  \ifx\svgwidth\undefined%
    \setlength{\unitlength}{203.06787205bp}%
    \ifx\svgscale\undefined%
      \relax%
    \else%
      \setlength{\unitlength}{\unitlength * \real{\svgscale}}%
    \fi%
  \else%
    \setlength{\unitlength}{\svgwidth}%
  \fi%
  \global\let\svgwidth\undefined%
  \global\let\svgscale\undefined%
  \makeatother%
  \begin{picture}(1,0.22557258)%
    \put(0,0){\includegraphics[width=\unitlength,page=1]{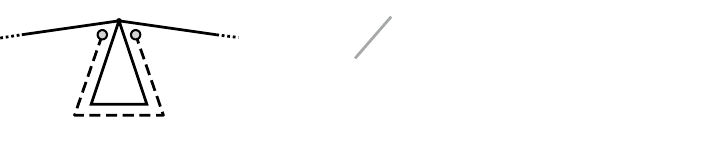}}%
    \put(0.09403438,0.16538865){\color[rgb]{0,0,0}\makebox(0,0)[lb]{\smash{$p_s$}}}%
    \put(0.20893672,0.16538865){\color[rgb]{0,0,0}\makebox(0,0)[lb]{\smash{$p_t$}}}%
    \put(0,0){\includegraphics[width=\unitlength,page=2]{dynamicprog.pdf}}%
    \put(0.50768916,0.19887505){\color[rgb]{0,0,0}\makebox(0,0)[lb]{\smash{$p_s$}}}%
    \put(0.43349221,0.05705054){\color[rgb]{0,0,0}\makebox(0,0)[lb]{\smash{$p_t$}}}%
    \put(0,0){\includegraphics[width=\unitlength,page=3]{dynamicprog.pdf}}%
    \put(0.45647476,0.13978147){\color[rgb]{0,0,0}\makebox(0,0)[lb]{\smash{$p_k'$}}}%
    \put(0.66884885,0.18311674){\color[rgb]{0,0,0}\makebox(0,0)[lb]{\smash{$p_c$}}}%
    \put(0,0){\includegraphics[width=\unitlength,page=4]{dynamicprog.pdf}}%
    \put(0.84649214,0.19887505){\color[rgb]{0,0,0}\makebox(0,0)[lb]{\smash{$p_s$}}}%
    \put(0.77229519,0.05705054){\color[rgb]{0,0,0}\makebox(0,0)[lb]{\smash{$p_t$}}}%
    \put(0,0){\includegraphics[width=\unitlength,page=5]{dynamicprog.pdf}}%
    \put(0.8697669,0.08856708){\color[rgb]{0,0,0}\makebox(0,0)[lb]{\smash{$p_{i'}$}}}%
    \put(0,0){\includegraphics[width=\unitlength,page=6]{dynamicprog.pdf}}%
    \put(0.95808963,0.0373527){\color[rgb]{0,0,0}\makebox(0,0)[lb]{\smash{$p_{j'}$}}}%
    \put(0.16095829,0.00207535){\color[rgb]{0,0,0}\makebox(0,0)[lb]{\smash{(a)}}}%
    \put(0.5352174,0.00207535){\color[rgb]{0,0,0}\makebox(0,0)[lb]{\smash{(b)}}}%
    \put(0.87008081,0.00207535){\color[rgb]{0,0,0}\makebox(0,0)[lb]{\smash{(c)}}}%
  \end{picture}%
\endgroup%
\caption{Examples of cases (ii) (iii) and (iv) of the dynamic programming.}
\label{fig:dynamicprog}
\end{figure}

\noindent\textbf{Correctness.}
We show that $A$ correctly computes $C[s,t]$ and that $C[1,n]$ corresponds to an optimal edge set $E'$.
The base case (i) is trivial.

\noindent
In case (ii), $p_s=p_t$ is a cut vertex of $W_F$, and all other vertices in $W[s,t]$ are descendants of $p_s$. Therefore,
there is no edge incident to a non-descendant of $p_s$ in $W_{s,t}$, and $p_s$ cannot be satisfied in $W_{s,t}$ (Fig.~\ref{fig:dynamicprog}(a)).

\noindent
In case (iii), since $p_s$ does not have descendants, the set $E'$ corresponding to $C[s,t]$
either has no edge incident to $p_s$ or it has an edge between $p_s$ and some descendant $p_k$ of a cut vertex $p_c$ in $W_{s,t}$. In the former case, we have $C[s,t]=C[s+1,t]$ since $E'$ satisfies all cut vertices relative to $W_{s+1,t}$.
In the latter, the edge $\{p_s,p_k\}$ creates a cycle satisfying the group of descendants containing $v_k$ for every cut vertex in $W_{c,k}$. It divides $F$ into two faces: $F_1$ (resp., $F_2$) whose facial walk contains $W_{s,k}$ (resp., $W_{k,t}$). Every group that still needs to be satisfied is either in $F_1$ or $F_2$.
If $E_1$ (resp., $E_2$) is the set with minimum weight between vertices in $W_{s,k}$ (resp., $W_{k,t}$)
that satisfies the groups of descendants in $F_1$ (resp., $F_2$), then $E'=E_1\cup E_2\cup\{\{s,k\}\}$.

\noindent
In case (iv), all non-descendants of $p_s$ in $W_{s,t}$ are in the set $\{p_{k+1},\ldots,p_t\}$, where $p_k$ is the last occurrence of $p_s$ in $W_{s,t}$. $E'$ must contain at least one edge in $X$ in order to satisfy $p_s$.
Let $\{p_i,p_j\}\in X$ be an edge that minimizes $i$, breaking ties by maximizing $j$.
Then, every edge $\{p_{i'},p_{j'}\}\in E'$ such that $(p_{i'},p_{j'})\in X$ is between vertices in $W_{i,j}$.
In particular, there exist no edge in $E'$ between a vertex in $W_{s,i}$ and $W_{j,t}$.
We can partition $E'$ into $E_1$, $E_2$, and $E_3$ such that they each contain only edges between vertices in $W_{s,i'}$, $W_{i',j'}$, and $W_{j',t}$ respectively, each of them being the minimum-weight set that satisfies the groups of descendants in their respective subproblems. The edges in $E_1$, $E_2$, $E_3$ cannot cross since they correspond to edge-disjoint walks in $W_F$. Thus, we have $C[s,t]=C[s,i']+C[i',j']+C[j',t']+f(i',j')$.

Let $E'$ be the edge set corresponding to $C[1,n]$. Since $E'$ satisfies all cut vertices in $W_{1,n}$, which corresponds to the cut vertices of $G_F$, the graph $(V_F,E_F\cup E')$ is 2-connected. Since $E'$ is a minimum-weight edge set that satisfies all groups of $W_{1,n}$ by definition, we have $\|E'\|=\|E_{OPT}\|$.

\noindent\textbf{Running time.}
The feasible edge weights $f(i,j)$ can be precomputed in $O(n^2)$ time~\cite{CW15,OW88} for all $i,j\in\{1,\ldots,n\}$.
There are $O(n^2)$ subproblems, and each can be computed in $O(n^2)$ time since it depends on $O(n^2)$ smaller subproblems.
Hence algorithm $A$ takes $O(n^4)$ time.

\begin{theorem}\label{thm:alg-PVCA}
For a connected PSLG $(V,E)$, an edge set $E^+$ of minimum  weight such that $(V,E\cup E^+)$ is 2-connected can be computed in $O(|V|^4)$ time.
\end{theorem}
\begin{proof}
We identify the faces of $(V,E)$ and run algorithm $A$ in all faces.
By Remark~\ref{rem:separation},
the union of the optimal solutions for all faces is the optimal solution for $G$.
Since the size of the union of all facial walks is $O(|V|)$, algorithm $A$ takes $O(|V|^4)$ time.
\end{proof}

\begin{theorem}\label{thm:alg-PECA}
For a connected PSLG $(V,E)$, an edge set $E^+$ of minimum weight such that $(V,E\cup E^+)$ is 2-edge-connected can be computed in $O(|V|^4)$ time.
\end{theorem}

\begin{proof}
	We modify algorithm $A$ into algorithm $B$, which computes the set $E^+$.
	Let an edge $\{p_c,p_{c+1}\}$ be a \emph{bridge relative to $W_{s,t}$}, $s\le c<t$, if it appears twice in $W_{s,t}$ (recall that an edge can only appear up to two times in a facial walk).
	Assume that $\{p_c,p_{c+1}\}$ and $\{p_{c'},p_{c'+1}\}$ are the first and second occurrences.
	Call $\{p_{c+1},\ldots,p_{c'}\}$ the descendants of $\{p_c,p_{c+1}\}$ in $W_{s,t}$,
    and all other vertices non-descendants. Replace the recurrence relation as follows.
	\begin{enumerate}
		\item [(i)] If $W_{s,t}$ does not contain a bridge relative to $W_{s,t}$, then $C[s,t]=0$.
		
		\item [(ii)] If $\{p_s,p_{s+1}\}$ is not a bridge relative to $W_{s,t}$, then $C[s,t]=\min\{C[s+1,t],$ $\min_{k\in\{s+2,\ldots,t-1\}}\{ C[s,k]+C[k,t]+f(s,k)\}\}$.
		
		\item [(iv)] If $\{p_s,p_{s+1}\}$ is a bridge relative to $W_{s,t}$, then let $X=\{$descendants of $\{p_s,p_{s+1}\}$ in $W_{s,t}\}\times
		\{\text{non-descendants of }$ $\{p_s,p_{s+1}\}$ in $W_{s,t}\}$.
		Set $C[s,t]=\min_{(p_i,p_j)\in X}\{ C[s,i]+C[i,j]+C[j,t]+f(i,j)\}$.
	\end{enumerate}
	The proof of correctness and runtime analysis are similar to algorithm $A$.
\end{proof}

Theorems~\ref{thm:alg-PVCA} and \ref{thm:alg-PECA} extend to any nonnegative weight function. In particular, a minimum cardinality edge set $E^+$ can also be computed in $O(|V|^4)$ time.

\section{Hardness of Connectivity Augmentation from 0 to 2}\label{sec:hardness}


\begin{theorem}\label{thm:hardness}
Given a (disconnected) PSLG $G=(V,E)$ and a positive integer $k$, deciding whether there exists an edge set $E^+$ such that $\|E^+\|\le k$ and $(V,E\cup E^+)$ is a 2-edge-connected PSLG is NP-hard.
\end{theorem}

\begin{proof}
	We reduce from \textsc{Planar-Monotone-3SAT}, which is NP-complete~\cite{BK12}.
	An instance of this problem is given by a plane bipartite graph between $n$ \emph{variables} and $m$ \emph{clauses} such that the variables are embedded on the $x$-axis, no edge crosses the $x$-axis and every clause has degree 2 or 3.
	A clause is called \emph{positive} if it is embedded on the upper half-plane and \emph{negative} otherwise.
	\textsc{Planar-Monotone-3SAT} asks if there is an assignment from $\{\texttt{true},\texttt{false}\}$ to variables such that each positive (resp., negative) clause is adjacent to at least one \texttt{true} (resp., \texttt{false}) variable.
	Given such an instance we build a PSLG $G=(V,E)$ as follows.
	We divide the reduction into variable, wire and clause gadgets.

	\begin{figure}[htp]
	\centering
		\def\svgwidth{.9\columnwidth}
\begingroup%
  \makeatletter%
  \providecommand\color[2][]{%
    \errmessage{(Inkscape) Color is used for the text in Inkscape, but the package 'color.sty' is not loaded}%
    \renewcommand\color[2][]{}%
  }%
  \providecommand\transparent[1]{%
    \errmessage{(Inkscape) Transparency is used (non-zero) for the text in Inkscape, but the package 'transparent.sty' is not loaded}%
    \renewcommand\transparent[1]{}%
  }%
  \providecommand\rotatebox[2]{#2}%
  \ifx\svgwidth\undefined%
    \setlength{\unitlength}{479.17813038bp}%
    \ifx\svgscale\undefined%
      \relax%
    \else%
      \setlength{\unitlength}{\unitlength * \real{\svgscale}}%
    \fi%
  \else%
    \setlength{\unitlength}{\svgwidth}%
  \fi%
  \global\let\svgwidth\undefined%
  \global\let\svgscale\undefined%
  \makeatother%
  \begin{picture}(1,0.67445368)%
    \put(0,0){\includegraphics[width=\unitlength,page=1]{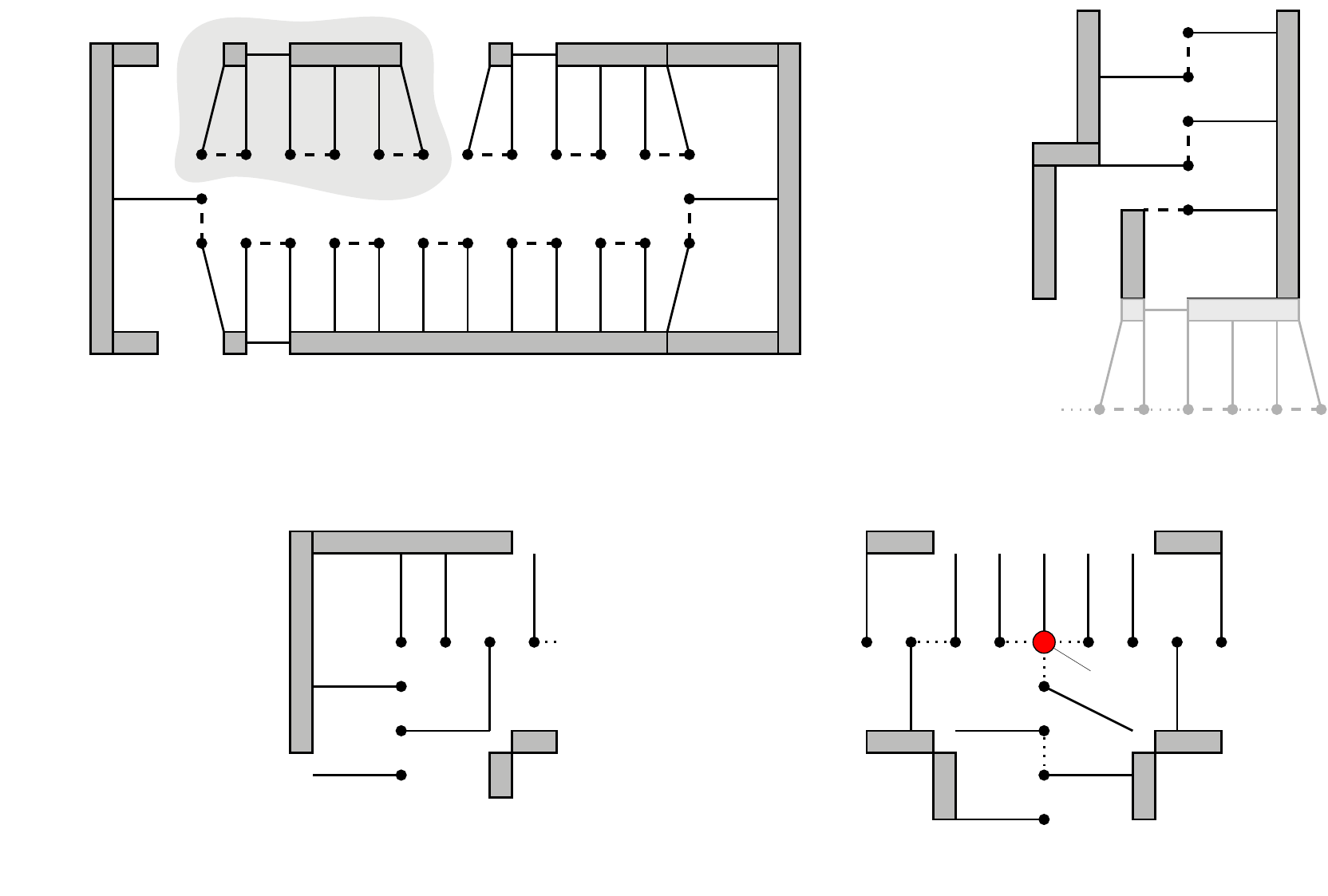}}%
    \put(0.15789859,0.66727207){\color[rgb]{0,0,0}\makebox(0,0)[lb]{\smash{$e_{i,r}$}}}%
    \put(0.7739534,0.396809){\color[rgb]{0,0,0}\makebox(0,0)[lb]{\smash{$e_{i,r}$}}}%
    \put(0,0){\includegraphics[width=\unitlength,page=2]{reduction.pdf}}%
    \put(0.29932747,0.3118339){\color[rgb]{0,0,0}\makebox(0,0)[lb]{\smash{(a)}}}%
    \put(0.87515493,0.3118339){\color[rgb]{0,0,0}\makebox(0,0)[lb]{\smash{(b)}}}%
    \put(0.29097977,0.00297177){\color[rgb]{0,0,0}\makebox(0,0)[lb]{\smash{(c)}}}%
    \put(0.77164448,0.00297177){\color[rgb]{0,0,0}\makebox(0,0)[lb]{\smash{(d)}}}%
    \put(0,0){\includegraphics[width=\unitlength,page=3]{reduction.pdf}}%
    \put(0.88375018,0.4835537){\color[rgb]{0,0,0}\makebox(0,0)[lb]{\smash{$l_{i,r}$}}}%
    \put(0,0){\includegraphics[width=\unitlength,page=4]{reduction.pdf}}%
    \put(-0.00070631,0.48362432){\color[rgb]{0,0,0}\makebox(0,0)[lb]{\smash{$e_{i}$}}}%
    \put(0,0){\includegraphics[width=\unitlength,page=5]{reduction.pdf}}%
    \put(0.6339468,0.48362432){\color[rgb]{0,0,0}\makebox(0,0)[lb]{\smash{$e_{i+1}$}}}%
    \put(0,0){\includegraphics[width=\unitlength,page=6]{reduction.pdf}}%
    \put(0.82544268,0.15465722){\color[rgb]{0,0,0}\makebox(0,0)[lb]{\smash{$c_{r}$}}}%
  \end{picture}%
\endgroup%
	\caption{Dotted and dashed lines represent \texttt{true} and \texttt{false} assignments respectively.
	Shaded rectangles represent cycles.
	(a) Variable gadget. (b) Connection between variable and wire gadgets. (c)~Turn in a wire gadget. (d) Clause gadget.}
	\label{fig:reduction}
	\end{figure}
	
	\noindent\textbf{Variable gadget.} A gadget of a variable adjacent to two positive and one negative clause is shown in Fig.~\ref{fig:reduction}(a).
	The gray boxes and small disks represent 2-connected components and leaves respectively.
	Each leaf requires at least one edge for the augmentation and the closest node from each of them is 1 unit apart.
	A pair of leaves can possibly share an edge and the $i$-th gadget for each $i\in\{1,\ldots,n\}$ has an even number of leaves $t_i$, hence the gadget requires at least $t_i/2$ length.
	There are exactly two possible ways to achieve this bound and they encode the \texttt{true}/\texttt{false} value of the variable.
	Fig.~\ref{fig:reduction}(a) can be generalized to other variables by repeating, omitting and changing the length and number of leaves of the component highlighted in the figure.
	
	\noindent\textbf{Wire gadget.} The wire gadget that connects the $i$-th variable to the $r$-th clause, denoted as the $(i,r)$-wire, share the edge $e_{i,r}$ with the $i$-th variable gadget and its first leaf is called $l_{i,r}$ (Fig.~\ref{fig:reduction}(b)).
	Turns, if needed, can be done as in Fig.~\ref{fig:reduction}(c).
	
	\noindent\textbf{Clause gadget.}
	The $r$-th clause gadget contains a special leaf called $c_r$ shown in red in Fig.~\ref{fig:reduction}(d) that is located at an odd distance $L_1$ from at most three leaves, say $l_{i,r}$, $l_{j,r}$, and $l_{k,r}$.
	If the clause is incident to two variables, use the turn shown in Fig.~\ref{fig:reduction}(c)
     as a clause gadget, naming the upper left leaf $c_r$.
	If the $r$-th clause is positive, we place all $l_{i,r}$ at an even distance from each other,
    which makes the position of $c_r$ always realizable.
	Use reflections through the $x$-axis for negative clauses.

	Let $t_{i,r}$ denote the number of leaves of the  $(i,r)$-wire.
	We set $$k=\sum_{i=1}^{n} t_i/2 + \sum_{\text{all }(i,r)\text{-wires}} (\lceil t_{i,r}/2\rceil +1).$$
	
	Assume that the \textsc{Planar-Monotone-3SAT} instance have a satisfying assignment.
	We build $E^+$ as follows.
	For each variable assigned \texttt{false} (resp., \texttt{true}) add the edges shown as dashed (resp., dotted) lines of Fig.~\ref{fig:reduction}(a) to $E^+$.
	For each $(i,r)$-wire add to $E^+$ all dashed lines shown in Fig.~\ref{fig:reduction}(b), (c) and the closest dashed line in Fig.~\ref{fig:reduction}(d) from the $(i,r)$-wire if the $i$-th variable is assigned \texttt{false} and the $r$-th clause is positive; or if the $i$-th variable is assigned \texttt{true} and the $r$-th clause is negative.
	Add the dotted lines otherwise.
	We obtain $E^+$ such that $\|E^+\|=k$ and the graph $(V,E\cup E^+)$ is 2-edge-connected.
	
Assume that there exists $E^+$ such that $\|E^+\|=k$ and the graph $(V,E\cup E^+)$ is 2-edge-connected.
	We call \emph{$(i,r)$-leaves} the set of all leaves in the $(i,r)$-wire and the two leaves in the clause gadget adjacent to it.
	Since the number of $(i,r)$-leaves is odd and the closest point from any such leaf is at least 1 unit away, the minimum length required to have at least an edge in $E^+$ incident to each leaf is $\lceil t_{i,r}/2\rceil +1$.
	Since $k$ is the sum of all such lower bounds, then: (i) the subset of $E^+$ that is incident to a $(i,r)$-leaf must have exactly $\lceil t_{i,r}/2\rceil +1$ unit length edges; and (ii) the subset of $E^+$ that is incident to the $i$-th variable gadget must be either the set of dashed or dotted lines in Fig.~\ref{fig:reduction}(a).
	Assume that $E^+$ contains the dotted lines in Fig.~\ref{fig:reduction}(a)  in the $i$-th variable gadget.
	Then, for a positive clause $r$, $E^+$ must contain an edge between $l_{i,r}$ and a point in the maximal 2-connected component of $G$ that contains $e_{i,r}$, or else the bridge in $G$ that connects such 2-connected component will remain a bridge in $(V,E\cup E^+)$.
	Therefore, all other $(i,r)$-leaves must be matched in order to satisfy (i).
	Since $(V,E\cup E^+)$ is 2-edge-connected, $c_r$ is connected to some other leaf.
	Then, if it is connected to a $(i,r)$-leaf, the $i$-th variable gadget uses the dotted edges (\texttt{true}).
	By applying the symmetric argument for negative clauses, all such clauses must be incident to a variable gadget using dashed edges (\texttt{false}).
	Then we have also a satisfying assignment for the \textsc{Planar-Monotone-3SAT} instance.
\end{proof}


\begin{corollary}
Given a (disconnected) PSLG $G=(V,E)$ and $k>0$, finding a set $E^+$ such that $\|E^+\|\le k$ and $(V,E\cup E^+)$ is a 2-connected PSLG is NP-hard.
\end{corollary}
\begin{proof}
	The same reduction used in the proof of Theorem~\ref{thm:hardness} also works for 2-connectivity.
	Notice that the length required by leaves is the same and any $E^+$, $\|E^+\|\le k$, that augments $G$ to 2-edge-connected also achieves 2-connectivity.
\end{proof}

\section{Dynamic Plane Graphs}\label{sec:dynamic}

\begin{theorem}\label{thm:dynamic}
Let $G=(V, E)$ be a connected PSLG with $|V|\geq 3$ and
no three collinear vertices.
Then there exists a sequence of edge insertion and deletion operations that transforms $G$ into a planar straight-line cycle $G'=(V, E')$ such that $\|E'\|\leq 2\|\MST(V)\|$ and that every intermediate graph is a connected planar straight-line graph of weight at most $\|E\|+\|\MST(V)\|$.  These bounds are the best possible.
\end{theorem}

\begin{proof}
We prove the upper bound constructively. We construct and analyze the sequence of edge insertion and deletion operations in five phases described below.

\paragraph{Phase~1: From $G$ to a tree.}
Let $G_1=(V,E_1)$ be an arbitrary spanning tree of $G=(V,E)$.
Successively delete the edges in $E\setminus E_1$ in an arbitrary order.
The graph remains connected, and its edge length bounded by $\|E_1\|\leq \|E\|$.

\paragraph{Phase~2: From an arbitrary tree to a tree in the Delaunay triangulation.}
Triangulate $G_1$ arbitrarily, and denote the triangulation by $T_1$.

Lawson~\cite{Law72} showed that every triangulation of a point set $V$ can be transformed into the Delaunay triangulation of $V$, denoted $DT(V)$, through a sequence of edge flips. If an edge $ab$ in a triangulation is adjacent to two triangles, $\Delta abc$ and $\Delta adb$, where $\{a,b,c,d\}$ is in convex position, then an \emph{edge flip} replaces edge $ab$ by a new edge $cd$, and produces a new triangulation on $V$. Lawson~\cite{Law77} also showed that there is a sequence of so-called \emph{Delaunay flips}, where point $d$ lies in the interior of the circumcircle of $\Delta(abc)$.
The length of the flip sequence is $O(n^2)$; this bound is the best possible~\cite{Law77,HNU11},
but finding the minimum length of flip-sequence between two triangulations is APX-Hard~\cite{Pil14}.


Let $T_1,\ldots , T_m$ be a sequence of triangulations on $V$, where $T_m=DT(V)$ and $T_i$ is obtained from $T_{i-1}$ by an edge flip described above for $i=2,\ldots, m$.
We now describe how to construct a sequence of spanning trees $G_1,\ldots, G_m$, where $G_i$ is a spanning tree in the triangulation $T_i$; and $G_i$ is obtained from $G_{i-1}$ by an edge insertion followed by an edge deletion (such that the intermediate graph is also a connected PSLG). Suppose that $T_i$ is obtained from $T_{i-1}$ by an edge flip that replaces $ab$ by $cd$. We distinguish two cases:

\begin{enumerate}\itemsep -2pt
\item Edge $ab$ is not in $G_{i-1}=(V,E_{i-1})$. Then let $G_{i}=G_{i-1}$.
\item Edge $ab$ is in $G_{i-1}=(V,E_{i-1})$. Then let $G_{i}=(V,E_{i-1}\cup\{e\}\setminus\{ab\})$,
               where $e\in \{ac,bc,ad,bd\}$ as described below.
\end{enumerate}

Assume that $ab\in E_{i-1}$. Since edge $ab$ violates the Delaunay Condition, point $d$ lies in the interior of the circumcircle of $\Delta(abc)$. Consequently, $\measuredangle acb+\measuredangle bda > \pi$ (equality would imply that $a$, $b$, $c$, and $d$ are cocircular). Without loss of generality, assume $\measuredangle acb>\pi/2$, that is $\bigtriangleup abc$ is an obtuse triangle. By the sine theorem, $ab$ is the longest side of $\bigtriangleup abc$, and so $\max(\|ac\|,\|bc\|) < \|ab\|$.

The deletion of edge $ab$ disconnects the spanning tree $G_{i-1}$ into two trees in which $a$ and $b$ are in different components. Vertex $c$ is in one of the two components. If $c$ and $a$ are in the same component, then $c$ and $b$ are in different components and $G_{i}:=(V,E_{i-1}\cup\{bc\}\setminus\{ab\})$ is a spanning tree. Otherwise $c$ and $a$ are in different components, and  $G_{i}:=(V,E_{i-1}\cup\{ac\}\setminus\{ab\})$ is a spanning tree. Since the edge $ab$ was replaced by a shorter edge, $ac$ or $bc$, we have $\|E_i\|<\|E_{i-1}\|$. Note also that the intermediate graph, $(V,E_{i-1}\cup\{ac\}\setminus\{ab\})$ or $(V,E_{i-1}\cup\{ac\}\setminus\{ab\})$, is connected and its length is bounded above by $\|E_{i-1}\|+\diam(V)\leq \|E\|+\|\MST(V)\|$.

\paragraph{Phase~3: From an arbitrary tree in the Delaunay triangulation to $\MST(V)$.}
To transform $G_m$ into the $\MST(V)$, we add the $\MST(V)$ edges to $G_m$ one at a time as described below.
It is well known that $DT(V)$ contains the Euclidean spanning tree $\MST(V)$ of $V$ as a subgraph.
Let $e$ be an edge in $\MST(V)$ that is not in $G_m$.
Add $e$ to $G_m$ creating a connected graph $G_m'$ whose total weight is at most
$\|E_m\|+\|e\|\le\|E_m\|+\|\MST(V)\|$.
Since $G_m$ is a spanning tree, there exist exactly one cycle in $G_m'$ and it contains $e$.
Delete a longest edge of such cycle which results in a connected graph $G_{i+1}$ that weighs
at most $\|E_m\|$. By repeating this procedure at most $|V|$ times, we obtain $G_m^{(k)}=\MST(V)$.

\paragraph{Phase~4: From $\MST(V)$ to a weakly simple polygon $C$.}
Given a graph $G_0=\MST(V)$, pick an arbitrary edge $uv$ of the convex hull of $V$ that is not present in $G_0$.
Let $p$ be the unique path between $u$ and $v$ in $G_0$, where $\|uv\|\leq \|p\|$ by the triangle inequality.
Augment $G_0$ with the edge $uv$ into a PSLG $G_1$, and let $C_1$ be the planar straight-line cycle formed by $uv$ and $P$.

We apply a sequence of edge insertion and deletions to $G_1$. In each step $i$, we maintain a PSLG $G_i=(V,E_i)$, a weakly simple polygon $C_i=(V(C_i),E(C_i))$ whose edges are contained in $G_i$, and an ordering among the multi-edges of $C_i$ (any multi-edge of $C_i$ is present as a single edge in $G_i$), such that all edges induced by $V(C_i)$ are in $E(C_i)$ and $\|E(C_i)\| + \|E_i\setminus E(C_i)\|\leq 2\|E_0\|$  (where the weight of any multi-edges of $C_i$ are counted with multiplicity, all other edges in $E_i$ are counted once). The vertex set $V(C_i)$ will grow monotonically until $V(C_i)=V$.

Given $G_i$ and $C_i$, we construct $G_{i+1}$ and $C_{i+1}$ as follows. Assume that $V(C_i)\neq V$.
Since $G_i$ is connected, there is a vertex $y$ in $C_i$ adjacent to some vertex outside of $C_i$.
In the counterclockwise order of edges incident to $y$ there must exist at least two transitions between edges to
vertices in $V(C_i)$ and not in $V(C_i)$, that is, a pair of consecutive edges $xy$ and $yz$ such that $x\in V(C_i)$ and $z\notin V(C_i)$,i.e.,  $xy\notin V(C_i)$ and $yz\in V(C_i)$. Since no three vertices are collinear, at least one of such pairs forms a convex walk $(x',y,z')$. Without loss of generality, assume that $x'\in V(C_i)$ and $z'\notin V(C_i)$.
An example is shown in Fig.~\ref{fig:phase4-5}(left).

Construct $C_{i+1}$ from $C_i$ by replacing edge $x'y$ with the path $\geo(x',y,z')\cup z'y$.
Similarly, we construct $G_{i+1}$ from $G_i$ by adding the geodesic path $\geo(x',y,z')$
(if an edge of $\geo(x',y,z')$ is already present in $G_i$, we increment its multiplicity by one),
and then deleting (one copy of) the edge $x'y$, and any other edges that are induced by $V(C_{i+1})$
but not present in $E(C_{i+1})$.

By the triangle inequality we have $\|\geo(x',y,z')\|\le \|x'y\|+\|yz'\|$.
We then have $\|E(C_{i+1})\| + \|E_{i+1}\setminus E(C_{i+1})\|\leq \|E(C_i)\| + \|E_i\setminus E(C_i)\|$,
since we added $\geo(x',y,z')$ and removed (one copy of) $x'y$.
Each step strictly increased the number of vertices in $V(C_i)$.
Consequently, Phase~4 executes at most $|V|-1$ times and the resulting graph
$G_k=(V,E_k)$ contains the weakly simple polygon $C_k$ $\|E(C_k) \|\le 2\|\MST(V)\|$.


\begin{figure}[h]
\centering
\includegraphics[width=0.95\columnwidth]{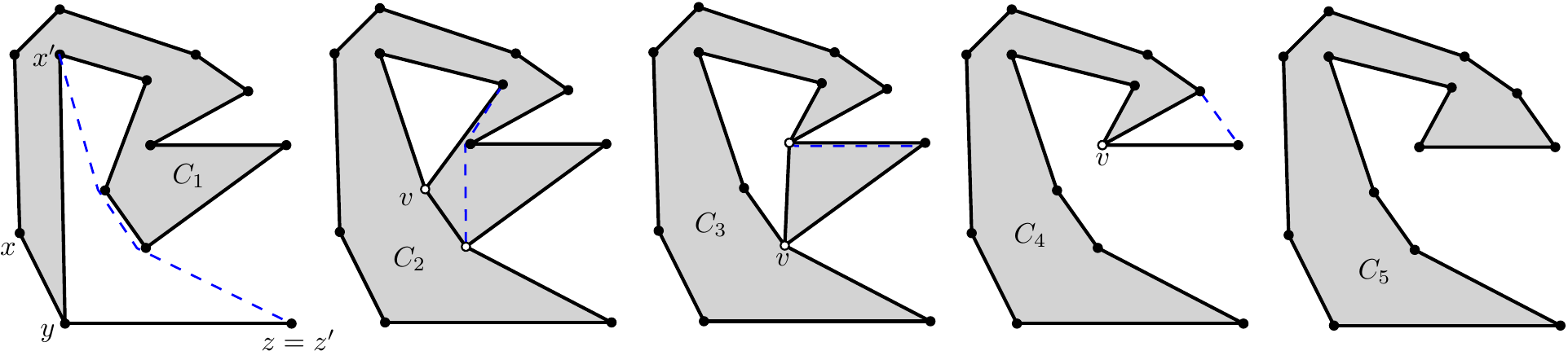}
\caption{Phases~4 and5 of our algorithm transforms a planar straight-line tree with 5 vertices (left) into a planar straight-line cycle (right). Dashed lines indicate path $(u,v,w)$ in the initial tree, and the paths $(x',y,z')$ in subsequent iterations. Empty circles indicate vertices of multiplicity two or higher in $C_i$.}
\label{fig:phase4-5}
\end{figure}

\paragraph{Phase~5: From a weakly simple polygon $C$ to a simple polygon $P$.}
Let $C_0$ be a weakly simple polygon on vertex set $V(C_0)=V$ and edge multiset $E(C_0)$. (Since no three vertices are collinear, no vertex lies in the relative interior of an edge). We construct a sequence of edge insertions and deletions that transform $C_0$ into a simple polygon on $V$, while monotonically decreasing the total edge length (counted with multiplicity) of intermediate weakly simple polygons $C_i$ until $C_i$ becomes Hamiltonian. This monotonicity ensures that the algorithm terminates.

We describe a generic step of the algorithm, where $C_i=(p_0,\ldots ,p_t)$, $p_t=p_0$, is a weakly simple polygon with vertex set $V(C_i)=V$. If every vertex has multiplicity one, then $C_i$ is a simple polygon, and our algorithm is complete. Otherwise, let $v\in V$ be a vertex of multiplicity two or higher. See Fig.~\ref{fig:phase4-5} for examples.
Since no three vertices are collinear, any two consecutive edges of $C_i$ define a convex and a concave angle $\angle p_{j-1} p_j p_{j+1}$ and $\angle p_{j+1} p_j p_{j-1}$.
Let $j$, $1\leq j\leq t$, be an index such that $p_j=v$ and $\min(\angle p_{j-1} p_j p_{j+1}, \angle p_{j+1} p_j p_{j-1})$ is minimal. Construct $C_{i+1}$ from $C_i$ by replacing the walk $(p_{j-1},p_j,p_{j+1})$ with $\geo(p_{j-1},p_j,p_{j+1})$.
The minimality of the angle guarantees that $C_{i+1}$ is a weakly simply polygon, and the triangle inequality yields
$\|\geo(p_{j-1},p_j,p_{j+1})\|<\|p_{j-1}p_j\|+\|p_jp_{j+1}\|$, as required.

\paragraph{Optimality.}
We now show our lower bounds, proving that the upper bound $2\|\MST(V)\|$ on the length of the output and $\|E\|+\|\MST(V)\|$ on the length of all intermediate graphs are the best possible.
For every $\varepsilon>0$, let $G_\varepsilon=(V,E)$ be a graph with 4 vertices $p_1=(0,0)$, $p_2=(0,\varepsilon)$, $p_3=(1,0)$, and $p_4=(1,\varepsilon)$, and edge set $E=\{p_1p_2, p_2p_3,p_3p_4\}$; depicted in Figure \ref{fig:3}. We have $\|E\| = \sqrt{1+\varepsilon^2}+2\varepsilon$. An $\MST$ of these four points is the path $(p_1,p_2,p_4,p_3)$, with $\|\MST(V)\|=1+2\varepsilon$.

The graph $G_\varepsilon$ is a path, which is not 2-connected. No edge of $G_\varepsilon$ can be deleted without disconnecting the graph, and only $p_1p_3$ or $p_2p_4$ can be inserted without introducing crossings, each of which has length 1. Hence the length of the 2nd graph in the sequence leading to 2-connectivity is $1+\sqrt{1+\varepsilon^2}+2\varepsilon$, which tends to $\|E\|+\|\MST(V)\|$ when $\varepsilon$ goes to 0.

Every 2-connected graph $G'=(V,E')$ contains at least two edges between $\{p_1,p_2\}$ and $\{p_3,p_4\}$, and the length of any edge between these vertex sets is at least 1. If $G'$ contains exactly two edges between $\{p_1,p_2\}$ and $\{p_3,p_4\}$, then it must contain the edges $p_1p_2$ and $p_3p_4$. Consequently, every 2-connected graph $G'=(V,E')$ satisfies $\|E'\|\geq 2+2\varepsilon$, and this bound is attained for the cycle $(p_1,p_2,p_4,p_3)$. We have $\lim_{\varepsilon \rightarrow 0}\|E'\|/\|E\|\geq 2$, and the ratio $\|E'\|/\|E\|\leq 2$ is the best possible.
\end{proof}

\end{document}